\documentclass[11pt]{article}

\usepackage[margin=1in]{geometry}
\usepackage{graphicx}
\usepackage[pdftex,colorlinks=true,linkcolor=blue,citecolor=blue,urlcolor=black]{hyperref}
\usepackage{amsmath, amsthm, amssymb}
\usepackage{subfigure}
\usepackage{comment}
\usepackage{url}
\usepackage{pdflscape}
\usepackage[ruled,lined,linesnumbered]{algorithm2e}
\usepackage{tikz}
\usepackage{calc}
\usepackage{multirow}
\usepackage{mathtools}

\newcommand{\Z}{\mathbb{Z}}

\newcommand{\E}{\mathbb{E}}

\newcommand{\ket}[1]{| #1 \rangle}

\DeclareMathOperator{\Var}{Var}

\newcommand{\be}{\begin{equation}}
\newcommand{\ee}{\end{equation}}
\newcommand{\bea}{\begin{eqnarray}}
\newcommand{\eea}{\end{eqnarray}}
\newcommand{\bes}{\begin{equation*}}
\newcommand{\ees}{\end{equation*}}
\newcommand{\beas}{\begin{eqnarray*}}
\newcommand{\eeas}{\end{eqnarray*}}

\makeatletter
\newtheorem*{rep@theorem}{\rep@title}
\newcommand{\newreptheorem}[2]{%
\newenvironment{rep#1}[1]{%
 \def\rep@title{#2 \ref{##1} (restated)}%
 \begin{rep@theorem}}%
 {\end{rep@theorem}}}
\makeatother

\newtheorem{thm}{Theorem}
\newtheorem*{thm*}{Theorem}

\newtheorem{lem}[thm]{Lemma}
\newtheorem{claim}[thm]{Claim}
\newtheorem*{lem*}{Lemma}

\newreptheorem{thm}{Theorem}
\newreptheorem{lem}{Lemma}

\newcommand{\boxalgm}[3]{
\renewcommand{\figurename}{Algorithm}
\begin{figure}[htb]
\noindent \framebox{
\begin{minipage}{\textwidth-0.5cm}
#3
\end{minipage}
}
\caption{#2}
\label{#1}
\end{figure}
\renewcommand{\figurename}{Figure}
}

\begin{document}

\title{The quantum complexity of approximating the frequency moments}
\author{Ashley Montanaro\thanks{School of Mathematics, University of Bristol, UK; {\tt ashley.montanaro@bristol.ac.uk}.}}
\maketitle

\begin{abstract}
The $k$'th frequency moment of a sequence of integers is defined as $F_k = \sum_j n_j^k$, where $n_j$ is the number of times that $j$ occurs in the sequence. Here we study the quantum complexity of approximately computing the frequency moments in two settings. In the query complexity setting, we wish to minimise the number of queries to the input used to approximate $F_k$ up to relative error $\epsilon$. We give quantum algorithms which outperform the best possible classical algorithms up to quadratically. In the multiple-pass streaming setting, we see the elements of the input one at a time, and seek to minimise the amount of storage space, or passes over the data, used to approximate $F_k$. We describe quantum algorithms for $F_0$, $F_2$ and $F_\infty$ in this model which substantially outperform the best possible classical algorithms in certain parameter regimes.
\end{abstract}

% ------------------------------------------------------------------------------

\section{Introduction}

Given a sequence of integers $a_1,\dots,a_n$, where $a_i \in [m] \coloneqq \{1,\dots,m\}$ for each $i$, let $n_j$ denote the number of elements in the sequence which are equal to the integer $j$. Then the $k$'th frequency moment is defined as
\[ F_k \coloneqq \sum_j n_j^k. \]
Thus, for example, $F_0$ is the number of distinct elements in the sequence, and $F_1=n$. We also define $F_\infty \coloneqq \max_j n_j$. Here we consider only integer $k$, and look to approximate $F_k$ up to relative error $\epsilon$ with bounded failure probability, or in other words to output $\widetilde{F_k}$ such that
\[ \Pr[|\widetilde{F_k} - F_k| > \epsilon F_k] \le 1/3. \]
As well as the intrinsic mathematical interest of this fundamental problem, it also has many practical uses, with $F_0$ and $F_2$ in particular occuring in database applications (see e.g.~\cite{alon99,charikar00}). It is therefore unsurprising that a vast amount of work, in a variety of different contexts, has been done to characterise the complexity of approximating the frequency moments; we summarise some of this below. In this work we address the complexity of approximating the frequency moments using a quantum computer.

We consider two different models where one could hope to achieve quantum speedups, both of which correspond to well-studied versions of the problem classically:

\begin{itemize}
\item The streaming model. In this model, we receive each item $a_i$ one at a time, in sequence. In the single-pass streaming model, we are asked to output an estimate for $F_k$ at the end of the sequence. In the multiple-pass streaming model, the stream repeats a number of times and we are asked to output an estimate after some number of repetitions. The challenge is that we assume that we only have access to limited storage space, and in particular not enough to store the whole stream.

\item The query complexity model. Here we can query arbitrary elements $a_i$, and seek to approximate $F_k$ using the minimal number of queries.
\end{itemize}
We assume in both cases that we know the total number of elements $n$ in advance. This seems to be essential in the query complexity model in order to enable queries to arbitrary elements of the input. In the streaming model, our algorithms will actually only need an upper bound on $n$. We allow probability of failure $1/3$, which can be improved to $\delta$, for arbitrary fixed $\delta > 0$, by repetition and taking the median.

In each of the above models, which we define somewhat more formally below, we obtain quantum improvements over the best possible classical complexities. Our main results can be summarised as follows:
\begin{itemize}
\item In the query complexity model, $F_0$ can be approximated with $O(\sqrt{n}/\epsilon)$ quantum queries, as compared with the classical lower bound of $\Omega(n)$ for $\epsilon = O(1)$~\cite{charikar00}, and this bound is tight.

\item In the query complexity model, $F_k$ can be approximated with $\widetilde{O}(n^{(1-1/k)(1-2^{k-2}/(2^k-1))}/\epsilon^2)$ quantum queries\footnote{The $\widetilde{O}$ notation hides polylogarithmic factors.} for $k \ge 2$, as compared with the classical lower bound of $\Omega(n^{1-1/k}/\epsilon^{2})$~\cite{acharya15}. Observe that $(1-2^{k-2}/(2^k-1)) \le 3/4$ for all $k \ge 2$, so this gives an asymptotic separation for all such $k$. In the important special case of $F_2$, the quantum upper bound is $\widetilde{O}(n^{1/3}/\epsilon^2)$, as compared with the classical lower bound $\Omega(n^{1/2}/\epsilon^{2})$, and the dependence on $n$ of this bound is tight up to logarithmic factors.

%\item The quantum query complexity of approximating $F_\infty$ is lower-bounded by $\Omega(n^{2/3})$.

\item In the streaming model, $F_0$ can be approximated by a bounded-error quantum algorithm which stores $O((\log n) \log 1/\epsilon)$ qubits and makes $O(1/\epsilon)$ passes over the input. Any classical algorithm that makes $T$ passes over the input must store $\Omega(1/(\epsilon^2 T))$ bits~\cite{chakrabarti12}, assuming that $\epsilon = \Omega(1/\sqrt{m})$.

\item In the streaming model, $F_2$ can be approximated by a bounded-error quantum algorithm which stores $O(\log n + \log(1/\epsilon))$ qubits and makes $\widetilde{O}(1/\epsilon)$ passes over the input. The classical lower bound is the same as for $F_0$~\cite{chakrabarti12}.

\item In the streaming model, $F_\infty$ can be computed exactly by a bounded-error quantum algorithm which stores $O(\log^2 n)$ qubits and makes $O(\sqrt{n})$ passes over the input. For sufficiently large $m$, any classical algorithm that makes $T$ passes must store $\Omega(n / T)$ bits~\cite{alon96}.

\item The above quantum upper bounds in the multiple-pass streaming model are all optimal, up to logarithmic factors.
\end{itemize}
For simplicity, we assume in these bounds that $m$ is quite large, $m \ge 2n$; more detailed complexities are given in the statements of the individual bounds below. We can also always assume that $m = O(n^2)$ because we can hash all the input elements to a set of size $O(n^2)$ without affecting the frequency moments, with 99\% probability; thus $\log m = \Theta(\log n)$. Finding efficient algorithms in terms of $m$ for smaller universe sizes $m$ has also been an important theme classically. However, the separations we obtain are generally maximised for large $m$.

Depending on $n$ and $\epsilon$, the quantum streaming complexities for $F_0$ and $F_2$ may not be substantially better, or could even be worse, than their corresponding classical lower bounds~\cite{chakrabarti12}. Further, the classical bounds are almost tight, as there exist single-pass streaming algorithms for these frequency moments which use $O(1/\epsilon^2 + \log m)$~\cite{kane10} and $O((\log m)/\epsilon^2)$~\cite{alon96} bits of storage, respectively. The quantum algorithms outperform the classical lower bounds in the regime $1/\sqrt{m} \ll \epsilon \ll 1/\log n$.

If we consider streaming algorithms which are restricted to making $\widetilde{O}(1/\epsilon)$ passes, classical algorithms for $F_0$ and $F_2$ must store $\widetilde{\Omega}(1/\epsilon)$ bits~\cite{chakrabarti12}, whereas the quantum algorithms store $O((\log n) \log 1/\epsilon)$ or $O(\log n + \log 1/\epsilon)$ qubits, respectively. For large $m$ and small $\epsilon$ (say $m = \Theta(n)$, $\epsilon = \Theta(1/n^\delta)$ for some $\delta > 1/2$), this is exponentially smaller than the best possible classical complexity. However, if we consider the product of the space usage $S$ and the number of passes $T$ (a standard measure used in time-space tradeoffs), the classical lower bound is $TS = \Omega(1/\epsilon^2)$, while the quantum upper bounds satisfy $TS = \widetilde{O}((\log n)/\epsilon)$. These could therefore be seen as near-quadratic separations.

The separations we obtain in the multiple-pass streaming model are, arguably, the first demonstration of a quantum advantage over classical computation for computing functions of practical interest in this model. Exponential separations have been shown in the one-pass streaming model by Gavinsky et al.~\cite{gavinsky09} for a partial function, and by Le Gall~\cite{legall06} for a total function. Unfortunately, these functions seem somewhat contrived. (However, it is possible to reinterpret the result of Le Gall as applying to computing the Disjointness function in the multiple-pass streaming model. In this setting the problem becomes more natural but the complexity reduction becomes only quadratic.)

% ------------------------------------------------------------------------------

\subsection{Related work}

There has been a huge amount of work characterising the classical complexity of approximating the frequency moments in various settings, only a fraction of which we mention here. See, for example,~\cite{braverman14,kane10} for further references.

In the streaming model:

\begin{itemize}
\item ($F_0$) Flajolet and Martin gave a single-pass streaming algorithm which uses $O(\log m)$ bits of space and computes $F_0$ up to a constant factor~\cite{flajolet85}. Alon, Matias and Szegedy improved this by replacing the randomness used in the Flajolet-Martin algorithm with a family of simple hash functions~\cite{alon96}. Bar-Yossef et al.\ gave several different algorithms for approximating $F_0$ up to a $(1+\epsilon)$ factor, using as little as $\widetilde{O}(1/\epsilon^2 + \log m)$ space~\cite{baryossef02}.
%A previous algorithm by Gibbons and Tirthapura achieved a similar space complexity~\cite{gibbons01}.
Kane, Nelson and Woodruff have now completed this line of research by giving a single-pass streaming algorithm which approximates $F_0$ using $O(1/\epsilon^2 + \log m)$ space~\cite{kane10}. This is optimal for single-pass streaming algorithms; a space lower bound of $\Omega(\log m)$ was shown by Alon, Matias and Szegedy~\cite{alon96}, and a lower bound of $\Omega(1/\epsilon^2)$ was shown by Woodruff~\cite{woodruff04}. This was generalised to an $\Omega(1/(\epsilon^2 T))$ lower bound for $T$-pass streaming algorithms by Chakrabarti and Regev~\cite{chakrabarti12}.

\item ($F_2$) Alon, Matias and Szegedy gave an $O((\log m)/\epsilon^2)$ single-pass streaming algorithm~\cite{alon96}, and also showed an $\Omega(\log m)$ lower bound. An $\Omega(1/\epsilon^2)$ lower bound for single-pass streaming algorithms was proven by Woodruff~\cite{woodruff04}, which was similarly extended to an $\Omega(1/(\epsilon^2 T))$ lower bound for $T$-pass streaming algorithms by Chakrabarti and Regev~\cite{chakrabarti12}.

\item ($F_k$, $k > 2$) Alon, Matias and Szegedy gave single-pass streaming algorithms using space $\widetilde{O}(m^{1-1/k})$~\cite{alon96}. An almost-optimal $\widetilde{O}(m^{1-2/k}/\epsilon^{10+4/k})$ algorithm was later given by Indyk and Woodruff for any $k > 2$~\cite{indyk05}. This was simplified by Bhuvanagiri et al., who also improved the dependence on $\epsilon$~\cite{bhuvanagiri06}. Very recently, Braverman et al.\ gave an $O(m^{1-2/k})$ algorithm for $k > 3$ and $\epsilon = \Omega(1)$~\cite{braverman14}. This effectively matches the tightest known general space lower bound on $T$-pass streaming algorithms, $\Omega(m^{1-2/k}/(\epsilon^{4/k} T))$ shown by Woodruff and Zhang~\cite{woodruff12}.

\item ($F_\infty$) Alon, Matias and Szegedy showed an $\Omega(m)$ space lower bound~\cite{alon96}, even for multiple-pass streaming algorithms with constant $\epsilon$, by a reduction from the communication complexity of Disjointness.
\end{itemize}
Near-optimal time-space tradeoffs for the related problem of exactly computing frequency moments over sliding windows were proven by Beame, Clifford and Machmouchi~\cite{beame13}.

The classical {\em query} complexity of approximating the frequency moments has also been studied, under the name of sample complexity. Charikar et al.~\cite{charikar00} gave a lower bound of $\Omega(n(1-\epsilon)^2)$ queries for approximating $F_0$. For any $k \ge 2$, Bar-Yossef~\cite{baryossef02a} showed a lower bound of $\Omega(n^{1-1/k} / \epsilon^{1/k})$, and a nearly matching upper bound (for $\epsilon = \Omega(1)$, $k = O(1)$) of $O(n^{1-1/k} / \epsilon^2)$. Very recently, the lower bound has been improved to a tight $\Omega(n^{1-1/k} / \epsilon^2)$ by Acharya et al.~\cite{acharya15}.

In the quantum setting, remarkably little seems to be known about the complexity of approximately computing frequency moments. Coffey and Prezkuta~\cite{coffey08} propose a quantum algorithm based on quantum counting which computes $F_\infty$ exactly for a sequence of $n$ elements, each picked from a set of size $m$, using $O(m\sqrt{n}\log m)$ queries. However, this complexity does not seem to be correct (cf.\ the lower bound of $\Omega(n)$ for exact computation of $F_\infty$ with $m=2$ which we prove below). %But one can indeed achieve a bound of this form for approximately computing $F_\infty$.

Kara~\cite{kara05} gave a quantum algorithm for finding an $\epsilon$-approximate modal value which uses $O((m^{3/2} \log m)/\epsilon)$ queries. Here a modal value is an element which occurs with frequency $F_\infty$ in the input sequence, an $\epsilon$-approximate modal value is an element which occurs with frequency at least $F_\infty/(1+\epsilon)$, and $m$ is again the size of the set of values. Note that once an approximate modal value is determined, $F_\infty$ itself can be approximately computed using quantum counting at the cost of an additional $O(\sqrt{n})$ queries. A quantum algorithm for computing $F_0$ over sliding windows was given in~\cite{beame13}, and achieves better time-space tradeoffs than are possible classically.

In terms of lower bounds, it was shown by Buhrman et al.~\cite{buhrman05} that computing $F_0$ exactly requires $\Omega(n)$ quantum queries. This result was later sharpened by Beame and Machmouchi~\cite{beame12} to show that even distinguishing between the cases of a function being 2-to-1 and almost 2-to-1 requires $\Omega(n)$ quantum queries.

Very recent independent work of Ambainis et al.~\cite{ambainis16} has considered a related problem to approximating $F_0$: testing the image size of a function. The quantum algorithm of~\cite{ambainis16} is based in the setting of property testing, and has subtly different parameters to the algorithm for $F_0$ presented here. Given oracle access to a function $f\colon[n] \rightarrow [m]$, their algorithm distinguishes between two cases: a) the image of $f$ is of size at most $k$; b) an $\epsilon$ fraction of the output values of $f$ need to be changed to reduce its image to size at most $k$. The algorithm uses $O(\sqrt{k / \epsilon} \log k)$ quantum queries.

% ------------------------------------------------------------------------------

\subsection{Techniques}

The new quantum algorithms we obtain are based on combining a number of different, previously known ingredients. Interestingly, ideas from classical streaming algorithms turn out to be useful for developing efficient quantum query algorithms; on the other hand, previously known efficient quantum query algorithms help to develop new quantum streaming algorithms.

The quantum query algorithm for $F_0$ is rather straightforward and is based around the idea of Bar-Yossef et al.~\cite{baryossef02} from the streaming setting that it suffices to compute the $O(1/\epsilon^2)$ smallest values of a pairwise independent hash function to estimate $F_0$. This can be done efficiently using a quantum algorithm of D\"urr et al.~\cite{durr04}. The algorithm for $F_k$, $k \ge 2$, is more involved, and starts from the observation~\cite{baryossef02a} that a good approximation of $F_k$ can be found by counting $k$-wise collisions in a large enough random subset $S$ of the inputs. On the other hand, if there are not too many $k$-wise collisions in $S$, the number of $k$-wise collisions can be computed efficiently using a quantum algorithm for $k$-distinctness, the problem of finding $k$ equal elements within $S$~\cite{belovs12}. The algorithm therefore runs the $k$-distinctness subroutine on random subsets $S$ which are exponentially increasing in size, until it finds a $k$-wise collision. It then switches to estimating the number of $k$-wise collisions using the $k$-distinctness algorithm.

In the quantum streaming model, to approximately compute $F_0$ we modify a different algorithm of Bar-Yossef et al.~\cite{baryossef02}. The idea is to use quantum amplitude estimation~\cite{brassard02} to approximate the probability that a random hash function $h\colon[m] \rightarrow [R]$, where $R = \Theta(F_0)$, maps any of the elements in the stream to 1. This enables a quadratic improvement, in terms of the scaling with $\epsilon$, over the classical algorithm in~\cite{baryossef02}. The main technical difficulty is to ensure that checking whether any of the elements in the stream are mapped to 1 can be implemented reversibly and space-efficiently. The efficient quantum algorithm for $F_2$ applies a quantum subroutine for efficient estimation of the expected value of random variables with bounded variance~\cite{montanaro15} to an estimator defined by Alon, Matias and Szegedy~\cite{alon96}. Finally, the algorithm for $F_\infty$ implements the quantum algorithm of D\"urr and H\o yer for finding the maximum~\cite{durr96} in a streaming setting.

The lower bounds in both the query and streaming models are based around the use of reductions. In the case of the query model, we reduce from well-studied problems in query complexity such as the threshold and element distinctness functions. In the case of the streaming model, we reduce from the Gap-Hamming, Disjointness and Equality problems in communication complexity.

\section{Quantum query complexity}

In this section, we describe quantum query algorithms for approximately computing $F_k$, followed by lower bounds. We use the standard model of quantum query complexity~\cite{buhrman02,hoyer05}. The algorithm is given access to the input via a unitary operator $O$ which maps $O\ket{i}\ket{x} \mapsto \ket{i}\ket{x + a_i}$, where $i \in [n]$, $x \in \Z_{m'}$ for some $m' \ge m$. The goal is to approximately compute $F_k$ with the minimal number of queries to $O$.

% ------------------------------------------------------------------------------

\subsection{$F_0$}

Our quantum algorithm for computing $F_0$ is based on a classical algorithm of Bar-Yossef et al.~\cite{baryossef02}. The starting point is the following idea of Flajolet and Martin~\cite{flajolet85} and Alon, Matias and Szegedy~\cite{alon96}: Given a uniformly random function $h\colon[m] \rightarrow [0,1]$, the value $\min_i h(a_i)$ should provide a good approximation of $1 / F_0$. Indeed, the expected minimum of $F_0$ random variables uniformly distributed in $[0,1]$ is precisely $1/(F_0+1)$. To achieve an estimate of $F_0$ accurate up to relative error $\epsilon$, it turns out to be sufficient to know the $O(1/\epsilon^2)$ smallest distinct values of $h(a_i)$, for a pairwise independent hash function $h$~\cite{baryossef02}. We can calculate these efficiently using a quantum algorithm of D\"urr et al.~\cite{durr04} for finding the $d$ smallest values of a function $f$, with the additional constraint that the values have to be of different types.

\begin{thm}[D\"urr et al.~\cite{durr04}]
\label{thm:dsmallest}
Given oracle access to two functions $f,g:[n] \rightarrow \Z$ and an integer $d$, there is a quantum algorithm which uses $O(\sqrt{dn})$ queries to $f$ and $g$ and outputs a set of $d'$ indices $I$, where $d' = \min\{d, |\{g(j):j \in [n]\}| \}$, such that:
\begin{itemize}
\item $g(i) \neq g(j)$ for all $i,j \in I$;
\item For all $i \in I$ and $j \in [n]\backslash I$, if $f(j) < f(i)$ then $f(i') \le f(j)$ for some $i' \in I$ with $g(i')=g(j)$.
\end{itemize}
The algorithm fails with probability at most $\delta$, for arbitrary $\delta = \Omega(1)$.
\end{thm}

The quantum algorithm for approximately computing $F_0$ is formally described as Algorithm \ref{alg:f0}.

\boxalgm{alg:f0}{Computing $F_0$}{
Set $d = \lceil 96 / \epsilon^2 \rceil$, $M = m^3$.
\begin{enumerate}
\item Let $h\colon[m] \rightarrow [M]$ be picked at random from a pairwise independent family of hash functions.
\item Use the algorithm of Theorem \ref{thm:dsmallest} with $f(i) = g(i) = h(a_i)$ and $\delta = 1/15$ to compute the $d$ smallest distinct values of $h(a_i)$. Let $v$ be the $d$'th smallest value.
\item Output $d M / v$.
\end{enumerate}
}

\begin{thm}
Algorithm \ref{alg:f0} makes $O(\sqrt{n}/\epsilon)$ queries and outputs an estimate of $F_0$ accurate up to relative error $\epsilon$ with probability at least $3/5-1/m$.
\end{thm}

\begin{proof}
It is shown in the proof of Theorem 1 in~\cite{baryossef02} that, for the value of $d$ chosen by Algorithm \ref{alg:f0}, $dM / v$ approximates $F_0$ up to a $1+\epsilon$ multiplicative factor with probability at least $2/3-1/m$. The claim then follows from the bounds of Theorem \ref{thm:dsmallest}.
\end{proof}

%The algorithm returns the $d$ smallest values within a list of $n$ integers in time $O(\sqrt{dn})$. (Actually, we need the variant of their algorithm which returns the $d$ smallest values of different types!) Alternatively, we can use an algorithm of Nayak and Wu which returns only the $d$'th smallest element in the same time complexity~\cite{nayak99} (er... up to log factors, and this isn't right because they have to be distinct). We apply this to the first algorithm from Bar-Yossef et al.~\cite{baryossef02}. The key point that we need about this algorithm is that, given a pairwise independent hash function $h$, the $t = \lceil 96 / \epsilon^2 \rceil$ smallest distinct values of $h(x_i)$ are sufficient to approximate $F_0$ up to a $1+\epsilon$ multiplicative factor. And, using D\"urr et al.'s algorithm, we can determine these in time $O(\sqrt{n}/\epsilon)$. We have shown the following theorem.

% ------------------------------------------------------------------------------

\subsection{$F_k$ for $k > 1$}

We begin by describing the technical tools required for the efficient quantum query algorithm for approximating $F_k$, starting with the underlying quantum subroutine for the so-called $k$-distinctness problem.

%\begin{thm}[Ambainis~\cite{ambainis04}]
%Fix $\delta$ such that $0 < \delta < 1$. There is a quantum algorithm which, given query access to a sequence $S= s_1,\dots,s_n$ integers, determines whether $S$ contains any duplicate elements. The output of the algorithm is either ``all distinct'' or a pair $i \neq j$ such that $s_i = s_j$. The algorithm uses $O(n^{2/3} \log(1/\delta))$ queries to $S$. It outputs an incorrect answer with probability at most $\delta$.
%\end{thm}

\begin{thm}[Belovs~\cite{belovs12}]
Fix integer $k \ge 2$ and real $\delta$ such that $0 < \delta < 1$. There is a quantum algorithm which, given query access to a sequence $S = s_1,\dots,s_n$, determines whether there exists a set $I$ of $k$ distinct indices such that $s_i=s_j$ for all $i,j \in I$. The output of the algorithm is either such a set $I$ or ``no''. The algorithm uses $O(n^{1-2^{k-2}/(2^k-1)} \log(1/\delta)) = o(n^{3/4} \log(1/\delta))$ queries to $S$. It outputs an incorrect answer with probability at most $\delta$.
\end{thm}

This algorithm is normally presented with failure probability $1/3$, but this can be reduced to $\delta$ by repetition, noting that we can check a claimed equal $k$-tuple using $k$ additional queries.

We will also need a technical lemma, which shows that $F_k$ can be expressed in terms of the number of $k$-wise collisions occurring in a random subset of the input integers. This lemma was essentially previously shown in~\cite{baryossef02a}, generalising the proof of~\cite{goldreich00} for the case $k=2$. However, as the terminology and parameter choice of these previous works is somewhat different to our usage here, we state and prove it afresh. Let $\binom{[\ell]}{k}$ denote the set of $k$-subsets of $[\ell]$.

\begin{lem}
\label{lem:collmeanvar}
Fix $\ell$ such that $1 \le \ell \le n$. Let $s_1,\dots,s_\ell \in [n]$ be picked uniformly at random and define
\[ C_k(s_1,\dots,s_\ell) \coloneqq |\{ T \in \binom{[\ell]}{k}: a_{s_i} = a_{s_j} \text{ for all } i,j \in T \}|. \]
Then
\[ \E_{s_1,\dots,s_\ell}[C_k(s_1,\dots,s_\ell)] = \frac{\binom{\ell}{k}F_k}{n^k} \]
and
\[ \Var(C_k) \le \sum_{q=k}^{2k-1} \left(\frac{\ell F_k^{1/k}}{n}\right)^q. \]
\end{lem}

\begin{proof}
See Appendix \ref{app:proofs}.
\end{proof}

\boxalgm{alg:fk}{Computing $F_k$}{
\begin{enumerate}
\item Set $\ell = n$.
\item For $i = 0,\dots,\lceil \log_2 n\rceil$:
\begin{enumerate}
\item Pick $s_1,\dots,s_{2^i} \in [n]$ uniformly at random and let $S$ be the sequence $a_{s_1},\dots,a_{s_{2^i}}$.
\item Apply a $k$-distinctness algorithm to $S$ with failure probability $1/(8\log_2 n)$.
\item If it returns a set of $k$ equal elements, set $\ell=2^i$ and terminate the loop.
\end{enumerate}
\item Set $M= \lceil K/\epsilon^2 \rceil$ for some universal constant $K$ to be determined. For $r = 1,\dots,M$:
\begin{enumerate}
\item Pick $s_1,\dots,s_\ell \in [n]$ uniformly at random and let $S$ be the sequence $a_{s_1},\dots,a_{s_\ell}$.
\item Let $T$ and $T'$ be empty sequences.
%\item Attempt to create a set $T$ of all indices $i \in \{1,\dots,\ell\}$ such that there exists $j \in \{1,\dots,\ell\}$ with $j \neq i$ and $a_{s_i} = a_{s_j}$, via $K$ iterations of the following subroutine:
%\item Repeat the following subroutine at most $\lceil 1/\epsilon \rceil$ times:
\item Repeat the following subroutine forever:
\begin{enumerate}
\item Apply a $k$-distinctness algorithm to $S\setminus T'$ with failure probability $\epsilon^2/(8K\ell)$.
\item If it returns ``no'', terminate.
\item If it returns a $k$-tuple of indices $I = (i_1,\dots,i_k)$ such that the corresponding elements of $S$ are all equal, update $T$ to $T \cup I$ and update $T'$ to $T' \cup \{i_k\}$.
\end{enumerate}
%\item Assume that $T$ now contains all elements in $S$ which are not unique in $S$.
\item Let the sequence $B = b_1,\dots,b_\ell$ be defined such that $b_j = S_{T_j}$ for $j=1,\dots, |T|$, and for each $j > |T|$, $b_j$ is an arbitrary integer distinct from all integers $b_{j'}$, $j' < j$.
%concatenating the elements of $S$ indexed by $T$ with a sequence of at most $\ell$ arbitrary distinct integers not in the set indexed by $T$.
\item Set $C^{(r)} \coloneqq |\{ U \in \binom{[\ell]}{k}: b_i = b_j \text{ for all } i,j\in U\}|$.
\end{enumerate}
\item Output $\frac{n^k}{M \binom{\ell}{k}}\sum_{r=1}^M C^{(r)}$.
\end{enumerate}
}

We are now ready to describe the algorithm for computing $F_k$, as Algorithm \ref{alg:fk}. Informally, the algorithm first uses the $k$-distinctness subroutine to determine a size $\ell$ such that a random subset of size $\ell$ contains some, but not too many, $k$-wise collisions (steps 1-2 below). This is already enough to compute $F_k$ up to constant multiplicative accuracy. The algorithm then switches to estimating the expected number of $k$-wise collisions in a random subsequence $S$ of length $\ell$ (steps 3-4), and uses this to approximate $F_k$ more precisely. The $k$-distinctness subroutine is used here too; by repeatedly running this subroutine (step 3c), we can find all subsets of size $k$ or greater such that all elements in the subset are equal. The total number of $k$-wise collisions will be invariant for any integer sequence consistent with the contents of such subsets. So once we know this information, we can compute the total number of $k$-wise collisions in $S$ without any further queries, by constructing an arbitrary sequence $B$ consistent with this (steps 3d-e).

\begin{thm}
Algorithm \ref{alg:fk} outputs an approximation of $F_k$ which is accurate up to relative error $1+\epsilon$ with probability at least $3/4$, using an expected number of queries which is
\[ O((n^{(1-1/k)(1-2^{k-2}/(2^k-1))}/\epsilon^2) \log (n/\epsilon)). \]
\end{thm}

\begin{proof}
First note that, by a union bound, we can assume that all the uses of the $k$-distinctness algorithms succeed, except with total error probability $1/4$. We now show that it is likely that $\ell$ is chosen such that $An / F_k^{1/k} \le \ell \le B n / F_k^{1/k}$, for some $A$ and $B$ relatively close to 1. By Markov's inequality and Lemma \ref{lem:collmeanvar},
\[ \Pr_{s_1,\dots,s_\ell}[C_k(s_1,\dots,s_\ell) \ge 1] \le \E_{s_1,\dots,s_\ell}[C_k(s_1,\dots,s_\ell)] = \frac{\binom{\ell}{k} F_k}{n^k} \le \frac{F_k}{n^k} \left(\frac{\ell e}{k} \right)^k. \]
Therefore, the probability that $\ell$ is set to be lower than $A n / F_k^{1/k}$ after the first loop is at most
\[ F_k\left(\frac{e}{kn}\right)^k  \sum_{i=0}^{\log_2(An/F_k^{1/k})} 2^{ik} \le 2 F_k\left(\frac{e}{kn}\right)^k \left(\frac{An}{F_k^{1/k}}\right)^k = 2 \left(\frac{A e}{k} \right)^k . \]
On the other hand, let $\ell = Dn/F_k^{1/k}$, for some $D$ such that $B \ge D \ge B/2 \ge 1$ and $\ell \ge 2$. Then from Lemma \ref{lem:collmeanvar},
\be \label{eq:varck} \Var(C_k) \le \sum_{q=k}^{2k-1} \left(\frac{\ell F_k^{1/k}}{n}\right)^q = \sum_{q=k}^{2k-1} D^q \le k\, D^{2k-1}. \ee
So, via Chebyshev's inequality (aka the second moment method), the probability that the algorithm fails to terminate at the point where $\ell = Dn/F_k^{1/k} \le Bn/F_k^{1/k}$ is at most
\[  \Pr_{s_1,\dots,s_\ell}[C_k(s_1,\dots,s_\ell) = 0] \le \frac{\Var(C_k)}{\E[C_k]^2} \le k\,\frac{D^{2k-1} n^{2k}}{\binom{\ell}{k}^2 F_k^2} \le k\,\frac{k^{2k} D^{2k-1} n^{2k}}{\ell^{2k} F_k^2} = \frac{k^{2k+1}}{D} \le \frac{2 k^{2k+1}}{B}. \]
Fixing, for example, $A = (k/e) 20^{-1/k}$, $B = 20 k^{2k+1}$, we have that $An/F_k^{1/k} \le \ell \le Bn/F_k^{1/k}$ except with  probability at most $1/5$. 

Assuming that $\ell$ is indeed bounded in this way, we now show that it suffices to repeat the second subroutine $O(1/\epsilon^2)$ times to estimate $F_k$ up to relative error $1+\epsilon$. By Lemma \ref{lem:collmeanvar} we have $\Var(C_k) \le k B^{2k-1}$. Assuming that the $k$-distinctness algorithm always succeeds, for all $r$, $C^{(r)}$ is equal to the number of $k$-wise collisions in a uniformly random subsequence $S$ of size $\ell$ of the input integers, so is distributed identically to $C_k$. Let $\overline{C} = \frac{1}{M} \sum_{r=1}^M C^{(r)}$. Then $\Var(\overline{C}) \le k B^{2k-1}/M$. By Chebyshev's inequality,
\[ \Pr\left[\left|\overline{C} - \frac{\binom{\ell}{k} F_k}{n^k} \right| \ge \epsilon \frac{\binom{\ell}{k} F_k}{n^k} \right] \le \frac{k^2 B^{4k-2}n^{2k}}{M\epsilon^2\binom{\ell}{k}^2 F_k^2} \le \frac{k^{2k+1} B^{4k-2}n^{2k}}{M\epsilon^2 \ell^{2k} F_k^2} \le \frac{k^{2k+1} B^{4k-2}}{A^{2k}M\epsilon^2}. \]
This implies that, in order to estimate $\E[C_k]$ up to relative error $1+\epsilon$ with failure probability at most $1/5$, say, it suffices to take $M = \lceil 5 k^{2k+1} B^{4k-2} /(A^{2k}\epsilon^2) \rceil$.

We finally compute the expected number of queries used by the algorithm. Assume that $\ell = O(n / F_k^{1/k}) = O(n^{1-1/k})$ and for conciseness write $\alpha = 1-2^{k-2}/(2^k-1)$. The first loop makes
\[ \sum_{i=0}^{\log_2 \ell} O(2^{\alpha i} \log \log n) = O(\ell^{\alpha} \log \log n) = O(n^{\alpha(1-1/k)} \log \log n) \]
queries. In the second loop, the expected number of repetitions of the subroutine is upper-bounded by the expected value of $C_k(s_1,\dots,s_\ell)$, because the number of elements such that there are at least $k-1$ other elements with the same value is a lower bound on the number of $k$-wise collisions. For $\ell \le Bn /  F_k^{1/k}$, this expected value is $O(1)$. Therefore, the expected number of queries used by the subroutine is $O((n/F_k^{1/k})^{\alpha} \log (\ell/\epsilon^2)) = O(n^{\alpha(1-1/k)} \log(n/\epsilon))$. As there are $O(1/\epsilon^2)$ uses of the subroutine, the overall expected number of queries is $O((n^{\alpha(1-1/k)}/\epsilon^2) \log (n/\epsilon))$ as claimed.
\end{proof}

We remark that it might be possible to improve the dependence on $\epsilon$ of this algorithm to $\widetilde{O}(1/\epsilon)$ by replacing step 3 with the use of a quantum algorithm for approximately computing the mean given a bound on the variance~\cite{montanaro15}, as in Section \ref{sec:f2streaming} below. The reason that this does not seem immediate is that we only know an upper bound on the expected runtime of step 3, rather than a worst-case bound as required by this quantum algorithm.

% ------------------------------------------------------------------------------

\subsection{$F_\infty$}
\label{sec:finfty}

We observe that $F_\infty$ is closely connected to the much-studied (and confusingly named) $k$-distinctness problem: determining whether a sequence $S$ of integers contains $k$ equal integers~\cite{ambainis04,belovs12}. Approximating $F_\infty$ up to relative error less than $1/(3k)$ allows one to solve $k$-distinctness. The case $k=2$ (element distinctness) has a lower bound of $\Omega(n^{2/3})$~\cite{ambainis04}, implying that the same lower bound holds for computing $F_\infty$ up to relative error $O(1)$. No stronger lower bound is known for higher $k$.

On the other hand, if we could solve $k$-distinctness for all $k$, we could compute $F_\infty$ exactly using binary search. One can show by straightforward techniques (see below) that $k$-distinctness requires $\Omega(n)$ queries for $k=\Omega(n)$. However, to approximate $F_\infty$ it is not necessary to solve $k$-distinctness exactly, but merely to solve a gapped version of $k$-distinctness. That is, we are given parameters $k$, $\epsilon$ and asked to distinguish between the following two cases:
\begin{enumerate}
\item $S$ contains $k$ equal elements;
\item $S$ contains no sequence of $(1-\epsilon)k$ equal elements or more.
\end{enumerate}
If we can approximate $F_\infty$ up to relative error $\epsilon$, we can clearly solve gapped $k$-distinctness. In addition, if we can solve gapped $k$-distinctness for arbitrary $k$, we can approximate $F_\infty$ using binary search, at a cost of an $O(\log n)$ factor in the number of queries used. % It therefore seems of interest to determine whether a sublinear bound could be found on the quantum complexity of gapped $k$-distinctness.

% ------------------------------------------------------------------------------

\subsection{Lower bounds}

We can obtain a number of easy lower bounds on the query complexity of estimating the frequency moments via reductions from previously studied problems. %We will use the following result of Nayak and Wu~\cite{nayak99}:

%\begin{thm}[Nayak and Wu~\cite{nayak99}]
%\label{thm:nayakwu}
%Let $f:[n] \rightarrow \{0,1\}$ either satisfy $|\{x:f(x)=1\}=n/2$, or $|\{x:f(x)=1\}|=(1+\epsilon)n/2$. Then any quantum algorithm that determines which is the case must make $\Omega(1/\epsilon)$ queries to $f$.
%\end{thm}

\begin{thm}
\label{thm:lb}
Assume $m \ge n+1$ and $\epsilon < 1/4$. The quantum query complexity of estimating $F_0$ up to relative error $\epsilon$, with failure probability at most $1/3$, is at least $\Omega(\sqrt{n/\epsilon})$. For any $k > 1$, the quantum query complexity of estimating $F_k$ up to relative error $\epsilon$, with failure probability at most $1/3$, is at least $\Omega(n^{1/2-1/(2k)}/\epsilon)$, and also obeys the bound $\Omega(n^{1/3})$. For $k=\infty$, the quantum query complexity is at least $\Omega(n^{2/3})$. For $k = \infty$, the quantum query complexity is also $\Omega(1/\epsilon)$ for any $m\ge 2$.
\end{thm}

\begin{proof}
We first deal with the $F_0$ lower bound, by a reduction from the threshold function $\operatorname{Th}_d$ on $n$ bits, for $d \le n/2$. This function is defined by $\operatorname{Th}_d(x) = 1$ if $|x| \ge d$, and $\operatorname{Th}_d(x) = 0$ otherwise. Given an input $x \in \{0,1\}^n$, define a sequence of $n$ integers $a_i$ by $a_i = i$ if $x_i = 1$, and $a_i = 0$ otherwise. Then $\operatorname{Th}_d(x) = 1$ if and only if $F_0 \ge d$. To determine this, it suffices to approximate $F_0$ up to relative error $1/(4d)$. As the threshold function has a lower bound of $\Omega(\sqrt{dn})$ for $d \le n/2$~\cite{beals01}, setting $d = \lceil 1/\epsilon \rceil$ implies the claimed result.

For the first bound for $k > 1$, we reduce quantum counting to estimating $F_k$. Consider the problem of determining whether an unknown $n$-bit string  has Hamming weight $\ell \le n/2$, or has Hamming weight $\ell + \Delta$, given access to queries to the bits of the string. This requires $\Omega(\sqrt{n/\Delta} + \sqrt{\ell n}/\Delta)$ quantum queries~\cite{nayak99}. For any bit-string $x\in \{0,1\}^n$, define a sequence of $n$ integers $a_i$ such that $a_i = i$ if $x_i = 0$, and $a_i = 0$ if $x_i = 1$. Consider two strings $x$, $y$ such that $|x| = n^{1/k}$, and $|y|=(n(1+8\epsilon))^{1/k}$; so $\ell = n^{1/k}$, $\Delta = n^{1/k}((1+8\epsilon)^{1/k}-1)$. Then $F_k$ is equal to $2n - n^{1/k}$ for the first corresponding sequence, and $2n+8\epsilon n - (n(1+8\epsilon))^{1/k}$ for the second sequence. One can verify that approximating $F_k$ up to relative error $\epsilon$ allows these two cases to be distinguished. Then the lower bound of~\cite{nayak99} implies that this problem requires
\[ \Omega\left(\frac{\sqrt{\ell n}}{\Delta}\right) = \Omega\left(\frac{n^{1/2-1/(2k)}}{(1+8\epsilon)^{1/k}-1} \right) = \Omega\left(\frac{n^{1/2-1/(2k)}}{\epsilon}\right) \]
quantum queries as claimed.
%For $k > 1$, we instead take $S_1 = 1,2,\dots,n$ and $S_2 = 0,0,\dots,0,1,2,\dots,n-d$, where $S_2$ has $d$ elements equal to 0 and $n-d$ distinct elements. Now the algorithm effectively has to find one of the 0 elements, requiring $\Omega(\sqrt{n/d}$ queries. We have $F_k(S_1) = n$ and $F_k(S_2) = d^k + n-d$. Thus the geometric mean is $\sqrt{(d^k+n-d)/n} \ge \sqrt{(d^k+n)/n}-1$. We take $d = (n(\Delta^2-1))^{1/k}$ for some $\Delta$. Then this bound is at least $\Delta-1$. We get a lower bound of $\Omega(n^{1/2-1/(2k)} \Delta^{1/k})$ as claimed.

For the $\Omega(n^{1/3})$ bound for $k > 1$, we use a reduction from the collision problem with small range, which has a lower bound of $\Omega(n^{1/3})$ quantum queries~\cite{kutin05,ambainis05}. Let $S_1$ be a sequence of $n$ numbers where each number occurs once, and let $S_2$ be a sequence of $n$ numbers where each number occurs twice. Then $F_k(S_1)=n$, $F_k(S_2) = n 2^{k-1}$. So estimating $F_k$ up to multiplicative error $\epsilon < 1/4$ allows these two cases to be distinguished for any $k>1$.

The $\Omega(n^{2/3})$ bound for $k=\infty$ follows from the lower bound on the quantum query complexity of element distinctness with small range~\cite{ambainis05}, which is clearly no easier, while the $\Omega(1/\epsilon)$ bound follows from~\cite{nayak99} using a similar argument to above, considering the problem of determining whether an unknown bit-string has Hamming weight $n/2$ or Hamming weight $(1+\epsilon)n/2$.
\end{proof}

We finally consider the case of computing $F_k$ with very high accuracy.

\begin{thm}
For any $k \neq 1$, and $m \ge n/2+1$, the quantum query complexity of computing $F_k$ up to $O(1/n)$ relative error, with probability of failure at most $1/3$, is $\Theta(n)$. For $k = \infty$, this bound still holds for any $m \ge 2$.
\end{thm}

\begin{proof}
First we take $k \neq \infty$. Beame and Machmouchi~\cite{beame12} showed that the quantum query complexity of the following problem is $\Theta(n)$. We are given query access to a sequence of $n$ integers $a_i$, each picked from $[m]$, where $m \ge n/2+1$. Either, for each $i$, there exists precisely one $j \neq i$ such that $a_i = a_j$; or this holds for precisely $n-2$ indices $i$, while for the other two indices $i'$, there is no $j \neq i'$ such that $a_{i'} = a_j$. Our task is to distinguish the two cases. Let $S_1$ be a sequence of the first form, and let $S_2$ be a sequence of the second form. For $k>1$, $F_k(S_1) = 2^{k-1} n$ and $F_k(S_2) = 2^{k-1}n + 2$. (For $k=0$, we have $n/2$ vs.\ $n/2-1$). So computing $F_k$ up to $O(1/n)$ relative error for any $k \neq 1$ allows us to distinguish $S_1$ and $S_2$.

For $k=\infty$, a lower bound of $\Omega(1/\epsilon)$ was already shown in Theorem \ref{thm:lb}.
\end{proof}

% ------------------------------------------------------------------------------

\section{Quantum streaming complexity}

We now move on to studying the quantum complexity of computing $F_k$ in the $T$-pass streaming setting. The model here is defined as follows:

\begin{enumerate}
\item The quantum algorithm stores a quantum state $\ket{\psi}$ of $S$ qubits, initialised to some starting state which does not depend on the input.
\item Integers in the input stream are received one-by-one until the end of the stream; as each integer $a$ arrives, a corresponding operation $U_a$ is performed on $\ket{\psi}$.
\item Step 2 is repeated $T$ times.
%\item Integers are received one-by-one; as each integer $a$ arrives, an operation $U_a$ is performed on the first register of $\ket{\psi}$, where $U_a \ket{x} = \ket{x + a}$. The quantum algorithm then performs some additional operations. Then $U_a^{-1}$ is performed. (Why? Because otherwise the quantum algorithm cannot even implement a classical streaming algorithm!) Alternatively -- do we just allow $U_a$ to be any operation the algorithm likes? This is at least as powerful.
\item At the end, a measurement is made on $\ket{\psi}$ which is supposed to reveal $F_k$.
\end{enumerate}
We remark that, in the case $T=1$, this model is very similar to a one-way quantum finite automaton~\cite{moore00}, and also to variants of the streaming model studied by Gavinsky et al.~\cite{gavinsky09} and Le Gall~\cite{legall06}. We could have (essentially equivalently) also defined this process by splitting the stored qubits into two registers, and performing the following operations for each arriving element $a$:

\begin{enumerate}
\item Apply $U_a$ to the first register, where $U_a \ket{x} = \ket{x + a}$.
\item Apply some fixed unitary operation $V$ to both registers.
\item Apply $U_a^{-1}$ to the first register.
\end{enumerate}
We assume that the algorithm knows the number $n$ of elements in the stream in advance.

%We assume that the final element in the stream is a special element $\perp$ not occurring elsewhere in the stream, to enable the algorithm to know when each repetition has terminated.

% ------------------------------------------------------------------------------

\subsection{$F_0$}

In order to obtain an efficient quantum algorithm for estimating $F_0$ in the multiple-pass streaming model, we will modify another efficient classical algorithm of Bar-Yossef et al.~\cite{baryossef02}. The basic idea is as follows. First, a rough estimate $R$ of $F_0$ can be obtained using a classical streaming algorithm of~\cite{alon96}, which returns some $R$ such that $R = \Theta(F_0)$ with high probability using only $O(\log m)$ space. Then, if $h\colon[m] \rightarrow [R]$ is a random function picked from a $t$-wise independent family of hash functions, for some $t = O(\log 1/\epsilon)$, the probability (over the random choice of $h$) that $h$ maps any of the elements in the stream to 1 provides a good estimate for $F_0$. Here, rather than sampling random functions $h$ to estimate this quantity, we will use amplitude estimation~\cite{brassard02}.

The main claim that we need to check is that the operation of checking whether $h$ maps any of the elements in the stream to 1 can be performed reversibly in a {\em space-efficient} fashion, with only a few passes over the stream. Note that standard reversible-computation techniques do not seem to immediately imply this, because the general technique used to reversibly implement each operation in an initially irreversible computation stores ``garbage'' bits for each step in the computation~\cite{nielsen00}, which are later uncomputed. Storing a garbage bit for each element in the stream would use space $\Theta(n)$.

%It seems that by modifying the 2nd algorithm of Bar-Yossef et al. we can get a multi-pass streaming quantum algorithm which stores only $O(\log m)$ qubits for any $\epsilon>0$. The idea is to apply quantum counting to the probability that a hash function picked from a certain set maps any of the elements in the stream to 0. To evaluate this union for all the hash functions (in superposition) takes $O(\log m)$ qubits (have to be careful\dots what about reversibility?), and we can do $O(1/\epsilon)$ rounds of quantum counting to estimate the probability up to $\epsilon$. This doesn't contradict the communication complexity lower bound because we have so many passes over the stream (is this legit?).

\begin{lem}
\label{lem:probest}
Let $a_1,\dots,a_n$ be a stream, and let $\mathcal{H}$, $|\mathcal{H}| = M$, be a family of functions $h_j:[m] \rightarrow [R]$ such that the map $H\colon(j,x) \mapsto h_j(x)$ can be implemented reversibly classically in time $T$ and space $S$. Then there is a quantum algorithm which estimates $\Pr_j [ \exists i, h_j(a_i) = 1 ]$ up to additive error $\epsilon$ using space $S + O(\log n + \log 1/\epsilon)$, $O(1/\epsilon)$ passes over the input, and time $O(nT)$ per pass.
\end{lem}

\begin{proof}
Set $p = \Pr_j [ \exists i, h_j(a_i) = 1 ]$, and for each $j$ write $N_j = | \{i: h_j(a_i)=1\}|$. We will apply quantum amplitude estimation to estimate $p$. To do so, we need to coherently implement the function $f(j) = [\exists i, h_j(a_i)=1]$. We will show that this can be implemented with two passes over the stream and space $S + O(\log n)$. For each element $a$, the map
\[ U_a : \ket{j}\ket{y} \mapsto \ket{j}\ket{y + [h_j(a)=1]} \]
can be implemented with two uses of $H$ (one to compute, and one to uncompute), each of which uses time $T$ and space $S$. After the whole stream has been read in, performing this map for each element $a_i$, we have effectively implemented the map
\[ \ket{j}\ket{0} \mapsto \ket{j}\ket{N_j} \]
in total time $O(nT)$. We now use an ancilla register to store whether the second register is nonzero:
\[ \ket{j}\ket{0}\ket{z} \mapsto \ket{j}\ket{N_j}\ket{z+[\exists i,h_j(a_i) = 1]} . \]
It remains to uncompute the contents of the second register. We can do this by reading the stream in again and performing the map
\[ U_a^{-1} : \ket{j}\ket{y} \mapsto \ket{j}\ket{y - [h_j(a)=1]}; \]
this requires only one extra qubit to remember whether we are adding or subtracting. The overall result is that we have implemented the map
\[ \ket{j}\ket{z} \mapsto \ket{j}\ket{z+[\exists i,h_j(a_i) = 1]} \]
as required, with two passes over the stream and $S + O(\log n)$ space. Quantum amplitude estimation requires $O(1/\epsilon)$ uses of this map and its inverse, and $O(\log 1/\epsilon)$ qubits of additional space, to estimate $\Pr_j [ \exists i, h_j(a_i) = 1 ]$ up to additive error $\epsilon$~\cite{brassard02}.
\end{proof}

We now apply Lemma \ref{lem:probest} to the framework of Bar-Yossef et al.~\cite{baryossef02}. Let $\mathcal{R}$ be the set of all functions $h\colon [m] \rightarrow [R]$, and set $r = \Pr_{h \in \mathcal{R}} [ \exists i, h(a_i) = 1 ]$. The following lemma says that, if we can approximate $r$, we can approximate $F_0$.

\begin{lem}[Corollary of Bar-Yossef et al.~\cite{baryossef02}]
\label{lem:by}
Fix $\epsilon \le 1$ and assume that $R$ satisfies $2 F_0 \le R \le 50 F_0$. Assume that $\widetilde{r}$ satisfies $|\widetilde{r} - r| \le \epsilon/150$ and define
\[ \widetilde{F_0} = \frac{\ln(1-\widetilde{r})}{\ln(1-1/R)}. \]
Then $|\widetilde{F_0} - F_0| \le \epsilon F_0$.
\end{lem}

In addition, if we replace $\mathcal{R}$ with a $t$-wise independent family of hash functions for large enough $t$, the probability (over the random choice of hash function $h$) that there exists an $i$ such that $h(a_i) = 1$ is not substantially affected.

\begin{lem}[Corollary of Bar-Yossef et al.~\cite{baryossef02}]
\label{lem:by2}
Let $\mathcal{H}$ be a $t$-wise independent family of hash functions $h_j:[m] \rightarrow [R]$, where $t = \lceil \ln (300/\epsilon)/\ln 5 \rceil$. Set $p = \Pr_j [ \exists i, h_j(a_i) = 1 ]$. Then $|p-r| \le \epsilon/300$.
\end{lem}

We now have all the ingredients we need for the $F_0$ estimation algorithm, which is formally described as Algorithm \ref{alg:baryossef}. Note that the only quantum ingredient of this algorithm is the use of amplitude estimation in step 3.

\boxalgm{alg:baryossef}{Streaming estimation of $F_0$, based on~\cite{baryossef02}}{
Set $t = \lceil \ln (300/\epsilon)/\ln 5 \rceil$.
\begin{enumerate}
\item Use the algorithm of \cite{alon96} to obtain an estimate $R$ such that $2 F_0 \le R \le 50 F_0$ using one pass over the stream and space $O(\log m)$, with probability at least $3/5$.
\item Let $\mathcal{H}$ be a family of $t$-wise independent hash functions $h_j: [m] \rightarrow [R]$.
\item Using the algorithm of Lemma \ref{lem:probest}, estimate $p = \Pr_j [ \exists i, h_j(a_i) = 1 ]$ up to additive error $\epsilon / 300$. Call the estimate $\widetilde{p}$.
\item Output $\ln(1-\widetilde{p})/\ln(1-1/R)$.
\end{enumerate}
}

\begin{thm}
Algorithm~\ref{alg:baryossef} computes $F_0$ up to relative error $\epsilon$, with failure probability at most $1/3$, using space $O(\log m \log(1/\epsilon) + \log n)$ and $O(1/\epsilon)$ passes over the input.
\end{thm}

\begin{proof}
The claim follows from Lemmas \ref{lem:probest}, \ref{lem:by}, and \ref{lem:by2}. In somewhat more detail: by Lemma \ref{lem:by2}, after step 3 of the algorithm, assuming that step 1 has succeeded, $|\widetilde{p}-p| \le \epsilon/300$ and $|p-r| \le \epsilon/300$, so $|\widetilde{p}-r| \le \epsilon/150$. By Lemma \ref{lem:by}, the output of the algorithm differs from $F_0$ by relative error $\epsilon$. Lemma \ref{lem:probest} states that approximating $p$ to the required level of accuracy can be done using $O(1/\epsilon)$ passes over the input. Step 1 uses space $O(\log m)$, and step 3 uses space $S + O(\log n + \log 1/\epsilon)$, where $S$ is the space required to specify a member of the family $\mathcal{H}$ of hash functions. $\mathcal{H}$ can be chosen such that $S = O(t \log m) = O(\log m \log(1/\epsilon))$, so the overall space usage is $O(\log m \log(1/\epsilon) + \log n)$ qubits.
\end{proof}

%Thus, if we approximate $p$ up to additive error $\epsilon / 300$ for a $t$-wise independent family of hash functions, we have approximated $r$ up to additive error $\epsilon/150$. Lemma \ref{lem:probest} says that we can do this using $O(1/\epsilon)$ passes over the input. Each hash function can be described by $O(t \log m)$ bits.

% ------------------------------------------------------------------------------

\subsection{$F_2$ and $F_k$, $k > 2$}
\label{sec:f2streaming}

To compute $F_2$ in the streaming model, we will apply the following result to ideas from the $F_2$ approximation algorithm of Alon, Matias and Szegedy~\cite{alon96}:

\begin{thm}[Quantum approximation with a bound on the relative variance~\cite{montanaro15}]
\label{thm:approxrel}
Let $v(\mathcal{A})$ be the distribution on the outputs of a quantum algorithm $\mathcal{A}$ such that $\E[v(\mathcal{A})^2]/\E[v(\mathcal{A})]^2 \le B$, for some $B \ge 1$, and $\mathcal{A}$ uses $S$ qubits of space. Then there is a quantum algorithm which approximates $\E[v(\mathcal{A})]$ up to additive error $\epsilon\, \E[v(\mathcal{A})]$, with probability at least $2/3$, and uses $\mathcal{A}$ and $\mathcal{A}^{-1}$ $O((B/\epsilon) \log^{3/2}(B/\epsilon) \log \log(B/\epsilon))$ times. The algorithm uses $O(S + \log(B/\epsilon))$ qubits of space.
\end{thm}

The algorithm of~\cite{alon96} uses a set of $M = O(m^2)$ 4-wise independent hash functions $h_i:[m] \rightarrow \{\pm 1\}$, and approximately computes the expected value of the function $f(i) = \left( \sum_{j=1}^m h_i(j) n_j \right)^2$ over the random choice of hash function $h_i$. This is sufficient to approximate $F_2$:

\begin{claim}[Alon, Matias and Szegedy~\cite{alon96}]
If $i \in [M]$ is picked uniformly at random, $\E_i[f(i)] = F_2$ and $\Var(f) \le 2F_2^2$.
\end{claim}

Here we would like to apply the algorithm of Theorem \ref{thm:approxrel} to accelerate this procedure. To do so, we need to compute $f$ reversibly and space-efficiently. For each hash function $h$, we can compute
\[ \sum_{i=1}^n h(a_i) = \sum_{j=1}^m h(j) n_j \]
using one pass over the stream. Further, we can compute $f(i)$ reversibly for any $i$ using two passes and space $O(\log m + \log n) = O(\log n)$. We first perform the map
\[ \ket{i}\ket{0}\ket{y} \mapsto \ket{i}\ket{\sum_{j=1}^n h_i(a_j)}\ket{y+f(i)}, \]
using one pass over the stream. We then use a second pass over the stream to subtract $h_i(a_j)$ for each $j$, so the state of the second register is effectively unchanged and we have performed the map $\ket{i}\ket{y} \mapsto \ket{i}\ket{y+f(i)}$. To carry out the inverse operation, we do the same thing in reverse.

We can therefore apply Theorem \ref{thm:approxrel} to $f$, and obtain the following result:

\begin{thm}
Algorithm~\ref{alg:ams} computes $F_2$ up to relative error $\epsilon$, with failure probability at most $1/3$, using space $O(\log n + \log(1/\epsilon))$ and
\[ O((1/\epsilon)\log^{3/2}(1/\epsilon) \log \log (1/\epsilon)) = \widetilde{O}(1/\epsilon) \]
passes over the input.
\end{thm}

\boxalgm{alg:ams}{Streaming estimation of $F_2$, based on~\cite{alon96}}{
\begin{enumerate}
\item Let $\mathcal{H}$ be a family of $O(m^2)$ 4-wise independent hash functions $h_j: [m] \rightarrow \{\pm1\}$.
\item Apply the algorithm of Theorem \ref{thm:approxrel} with accuracy bound $\epsilon$ to the following subroutine:
\begin{enumerate}
\item Pick $h \in \mathcal{H}$ uniformly at random.
\item Output $\left( \sum_{i=1}^n h(a_i) \right)^2$.
\end{enumerate}
\end{enumerate}
}

In the case of other moments $F_k$, for fixed $k > 2$, we can do something similar (but less efficient), using a different estimator described by Alon, Matias and Szegedy~\cite{alon96}. For $i \in [n]$, let $N(i) = |\{j : j \ge i, a_i = a_j \}|$, and let $N_k(i) = n(N(i)^k - (N(i)-1)^k)$. Then the following lemma holds:

\begin{lem}[Alon, Matias and Szegedy~\cite{alon96}]
If $i \in [n]$ is picked uniformly at random, then
\[ \E_i[N_k(i)] = F_k,\;\;\;\;\Var(N_k) \le k m^{1-1/k} F_k^2. \]
\end{lem}

It is clear that $N_k(i)$ can be computed reversibly using two passes over the stream (one to compute $N(i)$ and one to uncompute it), using space $O(\log n)$. Using Theorem \ref{thm:approxrel}, we can approximate $F_k$ up to additive error $\epsilon F_k$ using
\[ O((m^{1-1/k} / \epsilon) \log^{3/2}(m^{1-1/k} / \epsilon)\log \log(m^{1-1/k} / \epsilon) ) = \widetilde{O}(m^{1-1/k} / \epsilon) \]
passes and space $O(\log n + \log(m^{1-1/k} /\epsilon))$. This is sometimes superior to the best classical algorithms~\cite{braverman14}, but only for very small $\epsilon \le m^{-1/k}$.

% ------------------------------------------------------------------------------

\subsection{$F_\infty$}

The quantum streaming algorithm for computing $F_\infty$ is straightforward, based on a streaming implementation of the maximum-finding algorithm of D\"urr and H\o yer~\cite{durr96}. Using one pass over the stream, we can implement the map
\[ \ket{j}\ket{x} \mapsto \ket{j}\ket{x\pm n_j} \]
for any $j$ simply by adding (or subtracting) 1 to $x$ each time we see an element with value $j$. We can use this as an oracle within the quantum algorithm of D\"urr and H\o yer, which outputs the maximum of $N$ integers, using quantum space $O(\log^2 N)$ and $O(\sqrt{N})$ queries~\cite{durr96}. This immediately gives the following result:

\begin{thm}
There is a quantum algorithm which computes $F_\infty$ exactly, with failure probability at most $1/3$, using space $O(\log^2 m)$ and $O(\sqrt{n})$ passes over the input.
\end{thm}

% ------------------------------------------------------------------------------

\subsection{Lower bounds}

Just as in the classical world, lower bounds on quantum communication complexity (see e.g.~\cite{buhrman10} for an introduction) can be used to lower-bound the quantum complexity of computing functions in the streaming model. Alice and Bob divide the input $a_1,\dots,a_n,b_1,\dots,b_n$ into two; Alice gets the first half $a_1,\dots,a_n$, Bob the second half $b_1,\dots,b_n$. If there is a streaming algorithm which computes $F_k$ using $T$ passes over the input and stores $S$ qubits, by simulating this algorithm Alice and Bob obtain a two-way quantum communication protocol which communicates $O(TS)$ qubits in total, has no prior shared randomness or entanglement, and computes $F_k$. If there is a single-pass streaming algorithm, this gives a one-way quantum communication protocol.

We first observe that a bound of Alon, Matias and Szegedy~\cite{alon96} extends to give a general $\Omega(\log n)$ space lower bound in the quantum setting, and an $\Omega(\sqrt{n})$ bound for $k = \infty$. Throughout this section we assume that $m \ge 2n$.

\begin{thm}
Assume there exists a protocol in the multi-pass quantum streaming model which stores $S$ qubits and uses $T$ passes to compute $F_k$ for a stream of $n$ elements up to relative error $1/8$, with failure probability $1/3$. Then:
\begin{itemize}
\item if $k \neq 1$, $TS = \Omega(\log n)$;
\item if $k=\infty$, $TS = \Omega(\sqrt{n})$.
\end{itemize}
\end{thm}

\begin{proof}
In the proof we use $\circ$ to denote the concatenation operation on integer sequences.
\begin{itemize}
\item ($k\neq 1$): Alice and Bob will embed the equality function on $\Theta(n)$ bits in their inputs. Choose a family of $2^{\Omega(n)}$ subsets of $[n]$ of size $n/4$ such that each pair of subsets has at most $n/8$ elements in common. Then, if Alice receives input $x$, Bob receives input $y$, they encode these as subsets $S_x$, $S_y$. If their strings are equal, $F_k(S_x \circ S_y) = n 2^{k-2}$; otherwise, $F_0(S_x \circ S_y) \ge 3n/8$ and $F_k(S_x \circ S_y) \le n/4 + n 2^{k-3}$. For any $k \neq 1$, there is at least a constant factor gap between the values of $F_k$ in these two different cases. In particular, approximating $F_k$ up to relative error $1/8$ allows equality of $x$ and $y$ to be tested. This has an $\Omega(\log n)$ quantum communication complexity lower bound~\cite{kremer95}.
\item ($k=\infty$): Alon, Matias and Szegedy~\cite{alon96} give a reduction from Disjointness, which we repeat here. Given sets $S_a$, $S_b$, Alice and Bob simply apply the streaming algorithm to the concatenation $S_a \circ S_b$. If there are any elements in common, $F_\infty(S_a \circ S_b) \ge 2$; otherwise, $F_\infty(S_a \circ S_b) = 1$. So computing $F_\infty$ up to relative error $\epsilon$, for any $\epsilon<1/3$, allows Alice and Bob to determine whether their sets are disjoint. This has a quantum communication complexity lower bound of $\Omega(\sqrt{n})$~\cite{razborov03}.
\end{itemize}
\end{proof}

There is also a straightforward bound on the complexity of exact computation.

\begin{thm}
Assume there exists a protocol in the multi-pass quantum streaming model which stores $S$ qubits and uses $T$ passes to compute $F_k$ exactly for a stream of $n$ elements, with failure probability $1/3$. Then, if $k \notin \{0,1,\infty\}$, $TS = \Omega(n)$.
\end{thm}

\begin{proof}
We reduce from the problem of computing the Hamming distance between two bit-strings $x,y \in \{0,1\}^n$, which has a quantum communication complexity lower bound of $\Omega(n)$~\cite{huang06}. Given $x$, Alice produces an $n$-element sequence by setting the $i$'th element to be $i + nx_i$, and Bob produces a similar sequence whose $i$'th element is $i + ny_i$. Then for any $k \notin \{0,1,\infty\}$, $F_k$ of the concatenated sequence is precisely $n 2^k - (2^k-1)d(x,y)$, so determining $F_k$ exactly enables the Hamming distance between $x$ and $y$ to be computed.
\end{proof}

We now prove a general quantum lower bound on approximating $F_k$ in the streaming model, based on a sequence of reductions, all following from previously known results. The key point is that good approximations to $F_k$ are known to give efficient protocols, in the setting of two-way communication complexity, for a problem known as Gap-Hamming-Distance~\cite{indyk03,woodruff04,chakrabarti12}. Let $\text{GHD}_{n,t}$ be the partial function defined on a subset of $\{0,1\}^n \times \{0,1\}^n$ as follows:
\[ \text{GHD}_{n,t}(x,y) = \begin{cases} 0 & \text{if } d(x,y) \le t - \sqrt{n}\\ 1 & \text{if } d(x,y) \ge t + \sqrt{n} \end{cases} \]
%Alice and Bob are each given an $n$-bit string. Their task is to determine whether the Hamming distance of their two strings is at most $t - \sqrt{n}$, or at least $t + \sqrt{n}$, given the promise that one of these is the case.
%
Then we have the following sequence of claims:

\begin{claim}[from Indyk and Woodruff~\cite{indyk03}, Woodruff~\cite{woodruff04}]
\label{claim:iy}
Fix $k \notin \{1,\infty\}$. Assume there is a protocol in the multi-pass quantum (resp.\ classical) streaming model which stores $S$ qubits (resp.\ bits) and uses $T$ passes to approximate $F_k$ up to relative error $\epsilon$, with failure probability $p$, for a stream of $n$ elements picked from a universe of size $m$. Then, if $\epsilon \ge 1/\sqrt{m}$, there is $\ell = \Theta(1/\epsilon^2)$ such that there is a quantum (resp.\ classical) protocol in the 2-way communication model for $\text{GHD}_{\ell,\ell/2}$, which has failure probability $p$ and uses $O(TS)$ qubits of communication.
\end{claim}

\begin{claim}[Chakrabarti and Regev~\cite{chakrabarti12}]
\label{claim:cr}
Fix $\alpha \in (0,1/2]$. Then, if there is a quantum (resp.\ classical) communication protocol for $\text{GHD}_{n,n/2}$ using $c$ qubits (resp.\ bits) of communication with failure probability $p$, there is a quantum (resp.\ classical) communication protocol for $\text{GHD}_{n,\alpha n}$ using $O(c)$ qubits (resp.\ bits) of communication with failure probability $p$.
\end{claim}

\begin{claim}[Razborov~\cite{razborov03}, see~\cite{klauck07a} for version here]
\label{claim:razborov}
Fix a $Q$-qubit quantum communication protocol on $n$-bit inputs $x$, $y$, with acceptance probabilities $P(x,y)$. Write $P(i) = \E_{|x|=|y|=n/4,|x \wedge y|=i} [P(x,y)]$. Then, for every $d \le n/4$, there exists a degree-$d$ polynomial $q$ such that $|P(i) - q(i)| \le 2^{-d/4 + 2Q}$ for all $i \in \{0,\dots,n/4\}$.
\end{claim}

\begin{claim}[Nayak and Wu~\cite{nayak99}]
\label{claim:nw}
Let $q:\{0,\dots,n\} \rightarrow [-1/3,4/3]$ be a degree-$d$ polynomial such that $q(a) \le 1/3$, $q(b) \ge 2/3$ for some $a,b \in \{0,\dots,n\}$. Let $c \in \{a,b\}$ be such that $|n/2 - c|$ is maximised, and let $\Delta = |a - b|$. Then $d = \Omega(\sqrt{n / \Delta} + \sqrt{c(n-c)}/\Delta)$.
\end{claim}

\begin{thm}
Assume there exists a protocol in the multi-pass quantum streaming model which stores $S$ qubits and uses $T$ passes to compute $F_k$ up to relative error $\epsilon \ge 1/\sqrt{m}$ for a stream of $n$ elements picked from a universe of size $m$, with failure probability $1/4$, for $k \neq 1$. Then $TS = \Omega(1/\epsilon)$.
\end{thm}

\begin{proof}
The theorem follows from Claims \ref{claim:iy} to \ref{claim:nw}. Assume there exists a protocol which stores $S$ qubits and uses $T$ passes to approximate $F_k$ up to relative error $\epsilon$, with success probability $3/4$. Then by Claim \ref{claim:iy} there is a protocol for $\text{GHD}_{\ell,\ell/2}$, where $\ell = \Theta(1/\epsilon^2)$, with the same success probability, using $O(TS)$ qubits of communication. By Claim \ref{claim:cr}, there is similarly a protocol for $\text{GHD}_{\ell,\ell/8}$ using $Q = O(TS)$ qubits of communication. Now consider the special case of this problem where the inputs $x$, $y$ are restricted to Hamming weight $\ell/4$. Then $d(x,y) = \ell/2 - 2|x \wedge y|$, so in the notation of Claim \ref{claim:razborov}, $P(i) \le 1/4$ for $i \ge \ell/8 + \sqrt{\ell}/2$, and $P(i) \ge 3/4$ for $i \le \ell/8 - \sqrt{\ell}/2$. Taking $d = O(Q)$, there is a degree-$d$ polynomial $q$ such that $q(i) \le 1/3$ for $i \ge \ell/8 + \sqrt{\ell}/2$, $q(i) \ge 2/3$ for $i \le \ell/8 - \sqrt{\ell}/2$. By Claim \ref{claim:nw}, we must have $d = \Omega(\sqrt{\ell})$. Thus $TS = \Omega(\sqrt{\ell}) = \Omega(1/\epsilon)$, completing the proof.
\end{proof}

%Finally, we observe that if we generalise slightly to a model where the number of elements in the stream is a priori unknown, computing $F_1$ (the number of elements in the stream) has an $\Omega(\log \log n)$ quantum space lower bound, because estimating this up to a constant factor means there are $\Theta(\log n)$ different possible outputs, hence $\Theta(\log \log n)$ bits of information being stored. By Holevo's theorem, this requires $\Omega(\log \log n)$ qubits of storage space.

% ------------------------------------------------------------------------------

\section{Outlook}

We have initiated the study of the quantum complexity of approximately computing the frequency moments. However, there are still many open questions to be resolved. In particular:
\begin{itemize}
\item What is the quantum query complexity of approximately computing $F_\infty$? As discussed in Section~\ref{sec:finfty}, this seems to be closely connected to a ``gapped'' version of the well-studied $k$-distinctness problem, whose precise complexity is still unknown. For each $k>2$, improved bounds on $k$-distinctness would also improve our algorithm for computing $F_k$.

\item What is the quantum streaming complexity of approximating $F_k$ for $k > 2$? Just as the efficient quantum algorithm for $F_2$ is based on the streaming algorithm of Alon, Matias and Szegedy~\cite{alon96}, it would be interesting to determine whether more recent streaming algorithms for $F_k$~\cite{indyk03,bhuvanagiri06,braverman14} could be used to obtain efficient quantum algorithms.
\end{itemize}
It would also be interesting to find a practically relevant problem demonstrating a separation between quantum and classical streaming complexity in the one-way model. This could involve proving a separation between one-way quantum and classical communication complexity for a total function, which is a major open problem.

% ------------------------------------------------------------------------------

\subsection*{Acknowledgements}

I would like to thank Rapha\"el Clifford for suggesting the question which inspired this paper, Alexander Belovs for sending me a copy of~\cite{ambainis16}, and Jayadev Acharya for pointing out~\cite{acharya15}. I would also like to thank two anonymous referees for their careful and helpful comments. This work was supported by EPSRC Early Career Fellowship EP/L021005/1.

% ------------------------------------------------------------------------------

\appendix

\section{Proofs of combinatorial bounds}
\label{app:proofs}

\begin{replem}{lem:collmeanvar}
Fix $\ell$ such that $1 \le \ell \le n$. Let $s_1,\dots,s_\ell \in [n]$ be picked uniformly at random and define
\[ C_k(s_1,\dots,s_\ell) \coloneqq |\{ T \in \binom{[\ell]}{k}: a_{s_i} = a_{s_j} \text{ for all } i,j \in T \}|. \]
Then
\[ \E_{s_1,\dots,s_\ell}[C_k(s_1,\dots,s_\ell)] = \frac{\binom{\ell}{k}F_k}{n^k} \]
and
\[ \Var(C_k) \le \sum_{q=k}^{2k-1} \left(\frac{\ell F_k^{1/k}}{n}\right)^q. \]
\end{replem}

\begin{proof}
First observe that, for any set $T \in \binom{[\ell]}{k}$,
\be \label{eq:fk} \Pr_{s_1,\dots,s_\ell} \left[ a_{s_i} = a_{s_j} \text{ for all } i,j \in T \right] = \frac{1}{n^k} \sum_{p_1,\dots,p_k=1}^n [a_{p_1} = \dots = a_{p_k}]  = \frac{F_k}{n^k}. \ee
For the expectation, we compute
\beas
\E_{s_1,\dots,s_\ell}[C_k(s_1,\dots,s_\ell)] &=& \E_{s_1,\dots,s_\ell} [ |\{ T \in \binom{[\ell]}{k}: a_{s_i} = a_{s_j} \text{ for all } i,j \in T \}|]\\
&=& \sum_{T \in \binom{[\ell]}{k}} \Pr_{s_1,\dots,s_\ell} \left[ a_{s_i} = a_{s_j} \text{ for all } i,j \in T \right]\\
&=& \frac{\binom{\ell}{k}F_k}{n^k}.
\eeas
%th
We now bound the variance. We have
\beas
&& \E_{s_1,\dots,s_\ell}[C_k(s_1,\dots,s_\ell)^2]\\
&=& \E_{s_1,\dots,s_\ell}\left[\left(\sum_{T \in \binom{[\ell]}{k}} [a_{s_i} = a_{s_j} \text{ for all } i,j \in T] \right)^2 \right]\\
&=& \sum_{T,U \in \binom{[\ell]}{k}} \Pr_{s_1,\dots,s_\ell}[a_{s_i} = a_{s_j} \text{ for all } i,j \in T, \text{ and } a_{s_p} = a_{s_q} \text{ for all } p,q \in U]\\
&=& \sum_{T,U \in \binom{[\ell]}{k}, T \cap U \neq \emptyset} \Pr_{s_1,\dots,s_\ell}[a_{s_i} = a_{s_j} \text{ for all } i,j \in T, \text{ and } a_{s_p} = a_{s_q} \text{ for all } p,q \in U]\\
&& + \sum_{T,U \in \binom{[\ell]}{k}, T \cap U = \emptyset} \Pr_{s_1,\dots,s_\ell}[a_{s_i} = a_{s_j} \text{ for all } i,j \in T, \text{ and } a_{s_p} = a_{s_q} \text{ for all } p,q \in U]\\
%\eeas
%\beas
&=& \sum_{T,U \in \binom{[\ell]}{k}, T \cap U \neq \emptyset} \Pr_{s_1,\dots,s_\ell}[a_{s_i} = a_{s_j} \text{ for all } i,j \in T \cup U]\\
&& + \sum_{T,U \in \binom{[\ell]}{k}, T \cap U = \emptyset}  \Pr_{s_1,\dots,s_\ell}[a_{s_i} = a_{s_j} \text{ for all } i,j \in T] \Pr_{s_1,\dots,s_\ell}[a_{s_i} = a_{s_j} \text{ for all } i,j \in U]\\
\eeas
\beas
&=& \sum_{q=k}^{2k-1} |\{T,U \in \binom{[\ell]}{k}:|T \cup U| = q\}| \frac{F_q}{n^q} + |\{T,U \in \binom{[\ell]}{k}:|T \cup U| = 2k\}| \left(\frac{F_k}{n^k} \right)^2,
\eeas
where the final equality is (\ref{eq:fk}). We have the rough bound that
\[ |\{T,U \in \binom{[\ell]}{k}:|T \cup U| = q\}| \le \ell^{2k-q} \ell^{2(q-k)} = \ell^q, \]
because we can specify $T$ and $U$ by picking the $|T \cap U| = 2k-q$ elements in their intersection, then the $q-k$ elements in each set $(T \cup U) \backslash T$, $(T \cup U) \backslash U$. We also have
\[ F_q^{1/q} = \left(\sum_j n_j^q\right)^{1/q} \le \left(\sum_j n_j^k \right)^{1/k} = F_k^{1/k} \]
for any $q \ge k$, by monotonicity of $\ell_p$ norms. Combining these two bounds,
\[  \E_{s_1,\dots,s_\ell}[C_k(s_1,\dots,s_\ell)^2] \le \sum_{q=k}^{2k-1} \left(\frac{\ell F_k^{1/k}}{n}\right)^q + \binom{\ell}{k}\binom{\ell-k}{k} \left(\frac{F_k}{n^k} \right)^2. \]
Therefore
\beas
\Var(C_k) &=&  \E_{s_1,\dots,s_\ell}[C_k(s_1,\dots,s_\ell)^2] - (\E_{s_1,\dots,s_\ell}[C_k(s_1,\dots,s_\ell)])^2\\
&\le& \sum_{q=k}^{2k-1} \left(\frac{\ell F_k^{1/k}}{n}\right)^q - \binom{\ell}{k}\left(\frac{F_k}{n^k} \right)^2 \left(\binom{\ell}{k} - \binom{\ell-k}{k} \right) \\
&\le& \sum_{q=k}^{2k-1} \left(\frac{\ell F_k^{1/k}}{n}\right)^q
\eeas
as claimed.
\end{proof}

% ------------------------------------------------------------------------------
% ------------------------------------------------------------------------------

\bibliographystyle{plain}
\bibliography{../../thesis}

\begin{thebibliography}{10}

\bibitem{acharya15}
J.~Acharya, A.~Orlitsky, A.~Suresh, and H.~Tyagi.
\newblock The complexity of estimating {R\'e}nyi entropy.
\newblock In {\em Proc. 26\textsuperscript{th} ACM-SIAM Symp. Discrete
  Algorithms}, pages 1855--1869, 2015.
\newblock {\tt arXiv:1408.1000}.

\bibitem{alon99}
N.~Alon, P.~Gibbons, Y.~Matias, and M.~Szegedy.
\newblock Tracking join and self-join sizes in limited storage.
\newblock In {\em Proc. eighteenth ACM SIGMOD-SIGACT-SIGART Symposium on
  Principles of Database Systems (PODS '99)}, pages 10--20, 1999.

\bibitem{alon96}
N.~Alon, Y.~Matias, and M.~Szegedy.
\newblock The space complexity of approximating the frequency moments.
\newblock In {\em Proc. 28\textsuperscript{th} Annual ACM Symp. Theory of
  Computing}, pages 20--29, 1996.

\bibitem{ambainis04}
A.~Ambainis.
\newblock Quantum walk algorithm for element distinctness.
\newblock In {\em Proc. 45\textsuperscript{th} Annual Symp. Foundations of
  Computer Science}, pages 22--31, 2004.
\newblock {\tt quant-ph/0311001}.

\bibitem{ambainis05}
A.~Ambainis.
\newblock Polynomial degree and lower bounds in quantum complexity: Collision
  and element distinctness with small range.
\newblock {\em Theory of Computing}, 1:37--46, 2005.
\newblock {\tt quant-ph/0305179}.

\bibitem{ambainis16}
A.~Ambainis, A.~Belovs, O.~Regev, and R.~de~Wolf.
\newblock Efficient quantum algorithms for (gapped) group testing and junta
  testing.
\newblock In {\em Proc. 27\textsuperscript{th} ACM-SIAM Symp. Discrete
  Algorithms}, pages 903--922, 2016.
\newblock {\tt arXiv:1507.03126}.

\bibitem{baryossef02a}
Z.~Bar-Yossef.
\newblock {\em The Complexity of Massive Data Set Computations}.
\newblock PhD thesis, University of California at Berkeley, 2002.

\bibitem{baryossef02}
Z.~Bar-Yossef, T.~S. Jayram, S.~R. Kumar, D.~Sivakumar, and L.~Trevisan.
\newblock Counting distinct elements in a data stream.
\newblock In {\em Proc. RANDOM'02}, pages 1--10, 2002.

\bibitem{beals01}
R.~Beals, H.~Buhrman, R.~Cleve, M.~Mosca, and R.~de~Wolf.
\newblock Quantum lower bounds by polynomials.
\newblock {\em J. ACM}, 48(4):778--797, 2001.
\newblock {\tt quant-ph/9802049}.

\bibitem{beame13}
P.~Beame, R.~Clifford, and W.~Machmouchi.
\newblock Element distinctness, frequency moments, and sliding windows.
\newblock In {\em Proc. 54\textsuperscript{th} Annual Symp. Foundations of
  Computer Science}, pages 290--299, 2013.
\newblock {\tt arXiv:1309.3690}.

\bibitem{beame12}
P.~Beame and W.~Machmouchi.
\newblock The quantum query complexity of {AC0}.
\newblock {\em Quantum Inf. Comput.}, 12(7\&8):670--676, 2012.
\newblock {\tt arXiv:1008.2422}.

\bibitem{belovs12}
A.~Belovs.
\newblock Learning-graph-based quantum algorithm for k-distinctness.
\newblock In {\em Proc. 53\textsuperscript{rd} Annual Symp. Foundations of
  Computer Science}, pages 207--216, 2012.
\newblock {\tt arXiv:1205.1534}.

\bibitem{bhuvanagiri06}
L.~Bhuvanagiri, S.~Ganguly, D.~Kesh, and C.~Seha.
\newblock Simpler algorithm for estimating frequency moments of data streams.
\newblock In {\em Proc. 17\textsuperscript{th} ACM-SIAM Symp. Discrete
  Algorithms}, pages 708--713, 2006.

\bibitem{brassard02}
G.~Brassard, P.~H{\o }yer, M.~Mosca, and A.~Tapp.
\newblock Quantum amplitude amplification and estimation.
\newblock {\em Quantum Computation and Quantum Information: A Millennium
  Volume}, pages 53--74, 2002.
\newblock {\tt quant-ph/0005055}.

\bibitem{braverman14}
V.~Braverman, J.~Katzman, C.~Seidell, and G.~Vorsanger.
\newblock An optimal algorithm for large frequency moments using
  {$O(n^{1-2/k})$} bits.
\newblock In {\em Approximation, Randomization, and Combinatorial Optimization.
  Algorithms and Techniques (APPROX/RANDOM 2014)}, pages 531--544, 2014.
\newblock {\tt http://drops.dagstuhl.de/opus/volltexte/2014/4721}.

\bibitem{buhrman10}
H.~Buhrman, R.~Cleve, S.~Massar, and R.~de~Wolf.
\newblock Non-locality and communication complexity.
\newblock {\em Rev. Mod. Phys.}, 82(1):665--698, 2010.
\newblock {\tt arXiv:0907.3584}.

\bibitem{buhrman02}
H.~Buhrman and R.~de~Wolf.
\newblock Complexity measures and decision tree complexity: a survey.
\newblock {\em Theoretical Computer Science}, 288:21--43, 2002.

\bibitem{buhrman05}
H.~Buhrman, C.~D{\"u}rr, M.~Heiligman, P.~H{\o }yer, F.~Magniez, M.~Santha, and
  R.~de~Wolf.
\newblock Quantum algorithms for element distinctness.
\newblock {\em SIAM J. Comput.}, 34(6):1324--1330, 2005.
\newblock {\tt quant-ph/0007016}.

\bibitem{chakrabarti12}
A.~Chakrabarti and O.~Regev.
\newblock An optimal lower bound on the communication complexity of
  {G}ap-{H}amming-{D}istance.
\newblock {\em SIAM J. Comput.}, 41(5):1299--1317, 2012.
\newblock {\tt arXiv:1009.3460}.

\bibitem{charikar00}
M.~Charikar, S.~Chaudhuri, R.~Motwani, and V.~Narasayya.
\newblock Towards estimation error guarantees for distinct values.
\newblock In {\em Proc.\ PODS'00}, pages 268--279, 2000.

\bibitem{coffey08}
M.~Coffey and Z.~Prezkuta.
\newblock A quantum algorithm for finding the modal value.
\newblock {\em Quantum Information Processing}, 7(1):51--54, 2008.

\bibitem{durr04}
C.~D{\"u}rr, M.~Heiligman, P.~H{\o }yer, and M.~Mhalla.
\newblock Quantum query complexity of some graph problems.
\newblock In {\em Proc. 31\textsuperscript{st} {I}nternational {C}onference on
  {A}utomata, {L}anguages and {P}rogramming (ICALP'04)}, pages 481--493, 2004.
\newblock {\tt quant-ph/0401091}.

\bibitem{durr96}
C.~D{\"u}rr and P.~H{\o }yer.
\newblock A quantum algorithm for finding the minimum, 1996.
\newblock {\tt quant-ph/9607014}.

\bibitem{flajolet85}
P.~Flajolet and G.~N. Martin.
\newblock Probabilistic counting algorithms for data base applications.
\newblock {\em J. Comput. Syst. Sci.}, 31:182--209, 1985.

\bibitem{gavinsky09}
D.~Gavinsky, J.~Kempe, I.~Kerenidis, R.~Raz, and R.~de~Wolf.
\newblock Exponential separation for one-way quantum communication complexity,
  with applications to cryptography.
\newblock {\em SIAM J. Comput.}, 38(5):1695--1708, 2009.
\newblock {\tt quant-ph/0611209}.

\bibitem{goldreich00}
O.~Goldreich and D.~Ron.
\newblock On testing expansion in bounded-degree graphs.
\newblock Technical Report TR00-020, Electronic Colloquium on Computational
  Complexity, 2000.

\bibitem{hoyer05}
P.~H{\o }yer and R.~\v{S}palek.
\newblock Lower bounds on quantum query complexity.
\newblock {\em Bulletin of the European Association for Theoretical Computer
  Science}, 87:78--103, 2005.
\newblock {\tt quant-ph/0509153}.

\bibitem{huang06}
W.~Huang, Y.~Shi, S.~Zhang, and Y.~Zhu.
\newblock The communication complexity of the {H}amming distance problem.
\newblock {\em Information Processing Letters}, 99(4):149--153, 2006.
\newblock {\tt quant-ph/0509181}.

\bibitem{indyk03}
P.~Indyk and D.~Woodruff.
\newblock Tight lower bounds for the distinct elements problem.
\newblock In {\em Proc. 44\textsuperscript{th} Annual Symp. Foundations of
  Computer Science}, pages 283--288, 2003.

\bibitem{indyk05}
P.~Indyk and D.~Woodruff.
\newblock Optimal approximations of the frequency moments of data streams.
\newblock In {\em Proc. 37\textsuperscript{th} Annual ACM Symp. Theory of
  Computing}, pages 202--208, 2005.

\bibitem{kane10}
D.~Kane, J.~Nelson, and D.~Woodruff.
\newblock An optimal algorithm for the distinct elements problem.
\newblock In {\em Proc.\ twenty-ninth {ACM SIGMOD-SIGACT-SIGART} symposium on
  principles of database systems (PODS'10)}, pages 41--52, 2010.

\bibitem{kara05}
A.~Kara.
\newblock A quantum algorithm for finding an {$\epsilon$}-approximate mode.
\newblock Master's thesis, University of Waterloo, 2005.

\bibitem{klauck07a}
H.~Klauck, R.~{\v{S}}palek, and R.~de~Wolf.
\newblock Quantum and classical strong direct product theorems and optimal
  time-space tradeoffs.
\newblock {\em SIAM J. Comput.}, 36(5):1472--1493, 2007.

\bibitem{kremer95}
I.~Kremer.
\newblock Quantum communication.
\newblock Master's thesis, Hebrew University, 1995.

\bibitem{kutin05}
S.~Kutin.
\newblock Quantum lower bound for the collision problem with small range.
\newblock {\em Theory of Computing}, 1:29--36, 2005.
\newblock {\tt quant-ph/0304162}.

\bibitem{legall06}
F.~{Le Gall}.
\newblock Exponential separation of quantum and classical online space
  complexity.
\newblock In {\em Proc. 18th ACM SPAA}, pages 67--73, 2006.
\newblock {\tt quant-ph/0606066}.

\bibitem{montanaro15}
A.~Montanaro.
\newblock Quantum speedup of {M}onte {C}arlo methods.
\newblock {\em Proc. Roy. Soc. Ser. A}, 471(2181):20150301, 2015.
\newblock {\tt arXiv:1504.06987}.

\bibitem{moore00}
C.~Moore and J.~Crutchfield.
\newblock Quantum automata and quantum grammars.
\newblock {\em Theoretical Computer Science}, 237(1--2):275--306, 2000.
\newblock {\tt quant-ph/9707031}.

\bibitem{nayak99}
A.~Nayak and F.~Wu.
\newblock The quantum query complexity of approximating the median and related
  statistics.
\newblock In {\em Proc. 31\textsuperscript{st} Annual ACM Symp. Theory of
  Computing}, pages 384--393, 1999.
\newblock {\tt quant-ph/9804066}.

\bibitem{nielsen00}
M.~A. Nielsen and I.~L. Chuang.
\newblock {\em Quantum Computation and Quantum Information}.
\newblock Cambridge University Press, 2000.

\bibitem{razborov03}
A.~A. Razborov.
\newblock Quantum communication complexity of symmetric predicates.
\newblock {\em Izvestiya of the Russian Academy of Science}, 67:145--159, 2003.
\newblock {\tt quant-ph/0204025}.

\bibitem{woodruff04}
D.~Woodruff.
\newblock Optimal space lower bounds for all frequency moments.
\newblock In {\em Proc. 15\textsuperscript{th} ACM-SIAM Symp. Discrete
  Algorithms}, pages 167--175, 2004.

\bibitem{woodruff12}
D.~Woodruff and Q.~Zhang.
\newblock Tight bounds for distributed functional monitoring.
\newblock In {\em Proc. 44\textsuperscript{th} Annual ACM Symp. Theory of
  Computing}, pages 941--960, 2012.
\newblock {\tt arXiv:1112.5153}.

\end{thebibliography}

\end{document}